\documentclass[10pt,twocolumn]{article}

\usepackage[letterpaper,left=0.70in,right=0.70in,top=0.72in,bottom=0.62in]{geometry}
\usepackage{amsmath,amssymb}
\usepackage{graphicx}
\usepackage{booktabs}
\usepackage{array}
\usepackage{microtype}
\usepackage{enumitem}
\usepackage{url}
\usepackage[T1]{fontenc}

\setlist[itemize]{leftmargin=1.15em,itemsep=1pt,topsep=2pt,parsep=0pt}
\setlist[enumerate]{leftmargin=1.45em,itemsep=0pt,topsep=2pt,parsep=0pt}

\newcommand{\indicator}{\mathbf{1}}

\newtheorem{observation}{Observation}
\newtheorem{proposition}{Proposition}

\title{\vspace{-1.7em}\textbf{Beyond Edge Cuts: Activity-Weighted Multicast\\Hypergraph Mapping\\for Spiking Neural Networks on Mesh NoCs}}
\author{Amirreza Khorasanian\\
Electrical Engineering Student, Department of Electrical and Computer Engineering\\
University of Tehran, Tehran, Iran\\
\texttt{akh5793@gmail.com}}
\date{Extended preprint, August 2026}

\begin{document}
\maketitle
\vspace{-1.2em}

\begin{abstract}
Mapping spiking neural networks (SNNs) onto neuromorphic many-core platforms is often formulated with graph partitioning and pairwise placement costs. That abstraction is convenient, but it does not match the physical communication event: one spike from a source neuron is delivered to a \emph{set} of postsynaptic destinations, and routes to several destinations can share mesh links. We present M-HySMap, a route-aware, activity-weighted multicast hypergraph mapping framework. Each source neuron induces a directed hyperedge to its postsynaptic fanout, weighted by profiled activity. The mapper starts from strong activity-aware graph/QAP seeds and then optimizes distinct destination-core fanout, the union of deterministic mesh routes, and link congestion. The central algorithmic observation is locality: moving one neuron can change only its own source-rooted hyperedge and the hyperedges of its predecessors. This permits exact incremental gain evaluation while caching every unaffected route contribution. We expose this combinatorial structure in detail, derive a conservative placement lower bound, and describe a portfolio of partition and placement neighborhoods that preserves the best incumbent. Across a 115-job evidence suite on Potjans-inspired recurrent SNNs and mesh NoCs from $4\times4$ to $6\times6$, plus a $7\times7$ stress case, M-HySMap reduces routed multicast hops by 10.6--19.6\% over Activity+QAP and 19.7--41.1\% over Edge+QAP. Incremental updates accelerate refinement by 4.7--12.7$\times$ while matching full recomputation to numerical precision.
\end{abstract}

\section{Introduction}
Spiking neural networks communicate through sparse, event-driven spikes rather than dense synchronous activations~[8, 10]. On a neuromorphic many-core system, neurons must be assigned to hardware cores and inter-core spikes must traverse an on-chip network. Mapping therefore couples two classical combinatorial decisions: \emph{partitioning}, which decides which neurons share a core, and \emph{placement}, which decides where each logical core sits on the physical mesh.

Most mapping flows begin with an ordinary directed graph: vertices are neurons and edges are synapses. Partitioning then minimizes a cut objective, and placement often minimizes a pairwise distance-weighted communication objective. This is attractive because mature graph-partitioning and quadratic-assignment machinery can be reused~[5, 3, 6]. The difficulty is that the graph edge is not the hardware communication object.

A spike is emitted once by a source neuron and delivered to a destination \emph{set}. If many postsynaptic targets lie on the same remote core, the network need not treat them as independent source-to-target packets. Moreover, routes to several destination cores can share links. Consequently, two mappings with identical synaptic edge cut can induce different destination fanout, routed-hop counts, and bottleneck loads. Figure~1 shows the simplest form of this mismatch.

M-HySMap changes the optimized communication object without discarding strong classical methods. It first obtains an activity-aware graph partition and a pairwise QAP placement, then refines that mapping under a route-aware multicast objective. The expanded treatment in this preprint emphasizes four algorithmic ideas:
\begin{itemize}
\item \textbf{A representation change.} A source neuron and all of its postsynaptic targets form one directed, activity-weighted hyperedge.
\item \textbf{A route-union objective.} Cost is assigned to the union of mesh links used by the multicast event, so shared path prefixes are counted once per source spike.
\item \textbf{Exact local dependency.} A move of neuron $v$ changes only the hyperedge rooted at $v$ and hyperedges rooted at predecessors of $v$; all other cached contributions remain valid.
\item \textbf{A portfolio over neighborhoods.} Partition moves, core-location swaps, and alternating joint schedules explore complementary local neighborhoods, and the final mapping is the best incumbent under one common objective.
\end{itemize}

\begin{figure*}[t]
\centering
\includegraphics[width=0.95\textwidth]{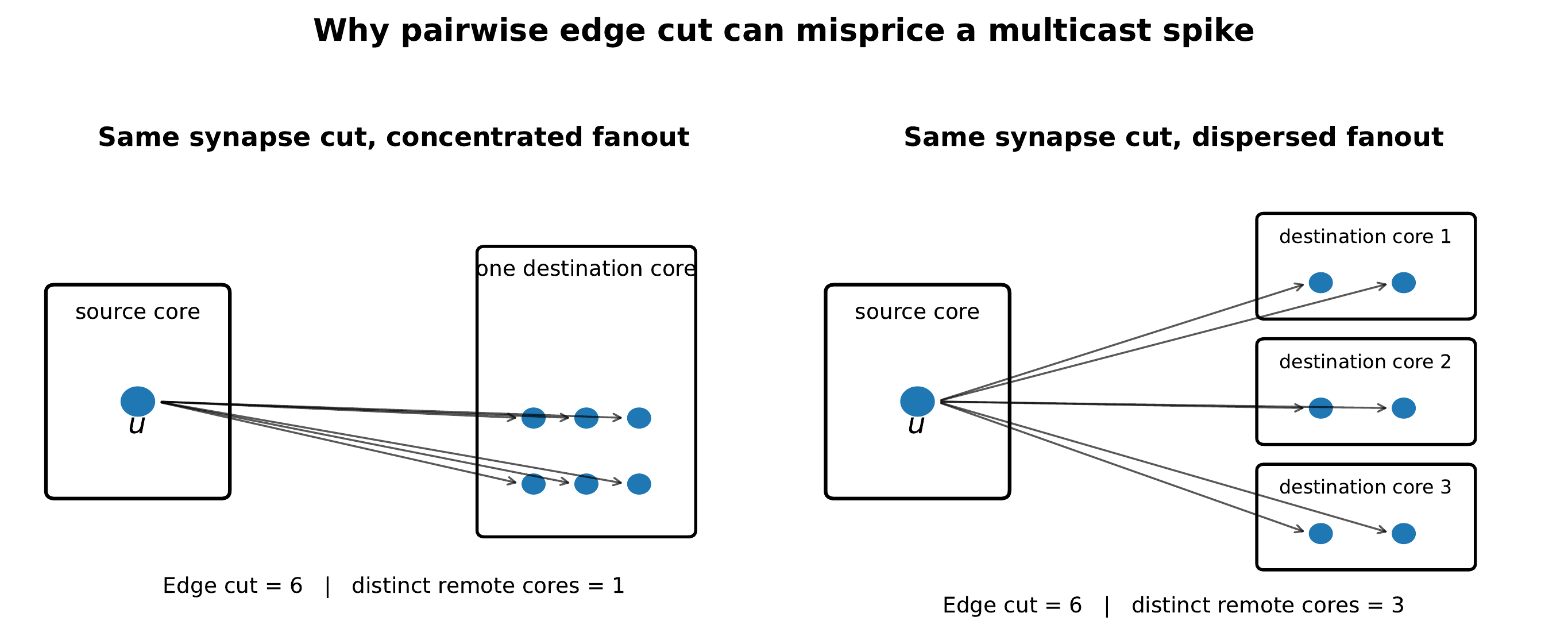}
\caption{Edge cut counts remote synapses independently. A multicast-aware view first collapses targets onto distinct destination cores; therefore equal edge cut need not imply equal network traffic.}
\end{figure*}

A shorter version of this work was prepared for ASP-DAC 2027. This extended preprint adds the combinatorial derivations, detailed search neighborhoods, incremental-update invariants, and worked examples that are difficult to fit within a conference page limit. It does not add unsupported hardware timing or energy claims.

\section{The Combinatorial Anatomy of SNN Mapping}
\subsection{Two coupled discrete spaces}
Let $G=(V,E)$ be a directed SNN graph with $n=|V|$ neurons. Let the hardware contain $k$ logical cores placed on a $R\times C$ mesh. A mapping has two components:
\begin{align}
p &: V\rightarrow\{0,\ldots,k-1\}, \tag{1}\\
\pi &: \{0,\ldots,k-1\}\rightarrow \mathcal{L}, \tag{2}
\end{align}
where $p$ assigns each neuron to a core and $\pi$ injectively assigns each logical core to a physical mesh location.

Ignoring balance and symmetries, the partition space contains $k^n$ labeled assignments. If exactly $k$ mesh locations are used, placement contributes $k!$ permutations. The raw joint space is therefore on the order of
\begin{equation}
k^n k!, \tag{3}
\end{equation}
which is already prohibitive for modest $n$ and $k$. The algorithm must therefore exploit structure rather than enumerate mappings.

Partitioning and placement are coupled. Moving a neuron can change which destination cores a source reaches; swapping two core locations can change the path lengths and shared links of many multicast events. M-HySMap addresses this coupling with alternating local neighborhoods rather than attempting a monolithic exact search.

\subsection{Why a graph cut can disagree with multicast traffic}
For a conventional directed edge cut,
\begin{equation}
C_{\mathrm{cut}}(p)=\sum_{(u,v)\in E}\indicator[p(u)\neq p(v)]. \tag{4}
\end{equation}
An activity-weighted variant replaces each unit contribution with the source activity $r_u$. Both objectives count remote synapses independently.

For source $u$, however, the physical first-order object is the set of distinct remote destination cores,
\begin{equation}
D_u(p)=\{p(v): v\in N^+(u),\ p(v)\neq p(u)\}. \tag{5}
\end{equation}
This is a set rather than a synapse count. If twenty targets land on one remote core, that core occurs once in $D_u$.

\begin{observation}[Many-to-one collapse]
For a fixed source $u$, edge cut depends on the number of remote target neurons, whereas multicast destination fanout depends on $|D_u|$. Therefore two mappings may have the same edge cut but different multicast packet replication and route cost.
\end{observation}

This distinction motivates the directed hyperedge
\begin{equation}
h_u=(u\rightarrow N^+(u)), \tag{6}
\end{equation}
weighted by profiled source activity $r_u\ge 0$. Recent work has independently argued for source-rooted hypergraph abstractions in SNN mapping~[14]; our focus is the next step: optimize the routed mesh links and congestion induced by that source-rooted event.

\begin{figure}[t]
\centering
\includegraphics[width=0.98\columnwidth]{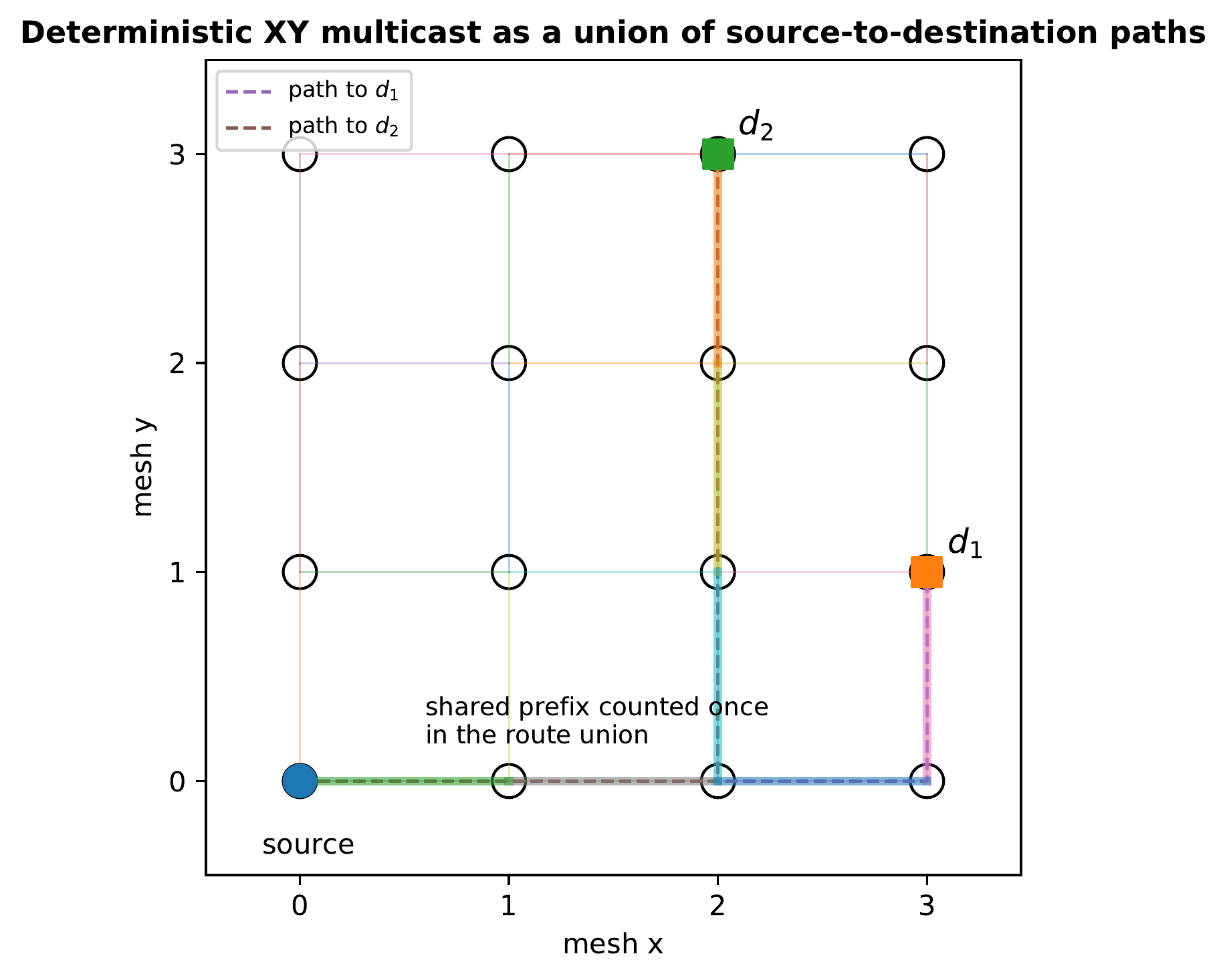}
\caption{Route-aware multicast cost uses the union of deterministic XY paths. Shared links are counted once per source spike. The union is the exact route set under the fixed routing policy; it is not claimed to be a minimum Steiner tree.}
\end{figure}

\section{Route-Aware Multicast Objective}
\subsection{From destination cores to routed-link unions}
Let $\mathrm{XY}(a,b)$ denote the deterministic X-then-Y Manhattan path between mesh locations $a$ and $b$. For source $u$, define the routed-link union
\begin{equation}
T_u(p,\pi)=\bigcup_{c\in D_u(p)} \mathrm{XY}\bigl(\pi(p(u)),\pi(c)\bigr). \tag{7}
\end{equation}
The union is important: if two paths share a prefix, those links appear once in $T_u$ for that source event. Figure~2 gives a concrete example.

The primary traffic term is
\begin{equation}
C_{\mathrm{hop}}(p,\pi)=\sum_{u\in V}r_u|T_u(p,\pi)|. \tag{8}
\end{equation}
For each directed mesh link $\ell$, the accumulated activity-weighted load is
\begin{equation}
L_\ell(p,\pi)=\sum_{u\in V}r_u\indicator[\ell\in T_u(p,\pi)]. \tag{9}
\end{equation}
We optimize
\begin{align}
\mathcal{J}(p,\pi)={}&\alpha C_{\mathrm{hop}}+\beta\max_\ell L_\ell \notag\\
&+\gamma\operatorname{Var}_\ell(L_\ell)+\delta_s I_s(p)+\delta_r I_r(p). \tag{10}
\end{align}
where $I_s$ and $I_r$ are optional size- and activity-imbalance terms. In the reported experiments, $\alpha=1$, $\beta=0.10$, $\gamma=0.01$, the soft imbalance weights are disabled, and a hard 15\% neuron-count slack is imposed on candidate moves.

The three active terms play different roles. $C_{\mathrm{hop}}$ measures total weighted link traversals under the chosen routing policy; $\max_\ell L_\ell$ discourages a single bottleneck link; and the variance term weakly discourages highly uneven load distributions. These are mapping-level traffic proxies, not calibrated cycle, throughput, or energy models.

\subsection{A small worked example}
Consider a $2\times2$ mesh and source $u$ on logical core 0 with $r_u=5$. Suppose four remote postsynaptic targets occupy cores 1 and 3, two on each. Edge cut counts four remote synapses. The destination-core set is only $D_u=\{1,3\}$. If the union of XY routes contains two distinct links, then $u$ contributes $5\times2=10$ routed hops. If a move colocates the core-3 targets with core 1, then $D_u=\{1\}$ and the route union may shrink to one link, giving contribution 5. The optimization is therefore sensitive to \emph{where fanout collapses}, not only to how many edges cross a partition boundary.

\subsection{A conservative placement lower bound}
The implementation maintains a simple diagnostic lower bound for the fixed-partition placement problem. Define the weighted remote fanout
\begin{equation}
F_{\mathrm{remote}}(p)=\sum_u r_u|D_u(p)|. \tag{11}
\end{equation}
Any connected route union spanning one source location and $d$ distinct destination locations requires at least $d$ links. Hence
\begin{equation}
C_{\mathrm{hop}}(p,\pi)\ge F_{\mathrm{remote}}(p). \tag{12}
\end{equation}
If the $R\times C$ mesh has
\begin{equation}
M=2\bigl(R(C-1)+C(R-1)\bigr) \tag{13}
\end{equation}
directed adjacent links, then $\max_\ell L_\ell\ge C_{\mathrm{hop}}/M\ge F_{\mathrm{remote}}/M$.

\begin{proposition}[Safe placement lower bound]
For fixed partition $p$,
\begin{equation}
\mathcal{J}(p,\pi)\ge \alpha F_{\mathrm{remote}}(p)+\beta\frac{F_{\mathrm{remote}}(p)}{M}+\delta_s I_s(p)+\delta_r I_r(p). \tag{14}
\end{equation}
\end{proposition}
\emph{Proof.} Equation~(12) lower-bounds the hop term. The maximum of $M$ nonnegative link loads is at least their average, and the sum of loads equals $C_{\mathrm{hop}}$. The variance term is nonnegative. The partition-only imbalance terms are fixed for fixed $p$. \hfill$\square$

This bound is intentionally conservative; it is useful as a sanity check and for tiny exact-placement experiments, not as a claim of tight optimality on the reported meshes.

\section{Algorithms}
Figure~3 summarizes the search architecture. A strong pairwise solution is treated as a valuable seed rather than a straw-man baseline. The algorithm then explores neighborhoods defined by the multicast objective.

\subsection{Activity profiling and source weights}
A short event-driven profiling run estimates
\begin{equation}
r_u=\frac{\text{recorded spikes of neuron }u}{T_{\mathrm{profile}}}. \tag{15}
\end{equation}
The main experiments use $T_{\mathrm{profile}}=0.20$ s, an external input rate of 6 Hz, and a cap of 150k processed events. The same activity vector is used by all activity-aware methods, so improvements cannot be attributed to different profiling traces.

\subsection{Strong graph partition and QAP placement seed}
The seed is intentionally classical. Starting from a balanced assignment, M-HySMap first performs greedy boundary-node refinement under edge cut and then activity-weighted edge cut. For a current partition $p$, the partition-to-partition flow matrix is
\begin{equation}
F_{ij}=\sum_{\substack{(u,v)\in E\\p(u)=i,\ p(v)=j}}w_u, \tag{16}
\end{equation}
where $w_u=1$ for Edge+QAP and $w_u=r_u$ for Activity+QAP. If $D_{ab}$ is Manhattan distance between mesh locations $a$ and $b$, the pairwise placement objective is
\begin{equation}
Q(\pi)=\sum_{i,j}F_{ij}D_{\pi(i),\pi(j)}. \tag{17}
\end{equation}
For tiny $k$, the code can enumerate placements exactly after fixing one symmetry. For the reported $k\ge16$ cases, it uses an anytime multi-start greedy 2-swap search. This produces the Activity+QAP seed $(p_0,\pi_0)$.

The role of QAP is worth emphasizing. Pairwise flow times physical distance is a strong proxy; M-HySMap does not discard it. Instead, the QAP solution becomes a warm start for the more faithful route-union objective.

\subsection{Boundary-node partition neighborhood}
A neuron is a \emph{boundary node} when at least one predecessor or successor lies on a \emph{different} core. Refinement concentrates on boundary nodes because moving a completely internal node is less likely to change inter-core structure.

For boundary neuron $v$, candidate destination cores are formed from the cores occupied by its incoming and outgoing neighbors, excluding its current core. If this candidate set is very small, a small number of additional cores are sampled to preserve limited exploration. A move $v:a\rightarrow b$ is feasible only when it respects hard balance bounds
\begin{equation}
\left\lfloor(1-s)\frac{n}{k}\right\rfloor\le |p^{-1}(c)|\le\left\lceil(1+s)\frac{n}{k}\right\rceil, \tag{18}
\end{equation}
with slack $s=0.15$ in the heavy evidence suite.

Within a pass, nodes are shuffled, all feasible candidate cores for a node are scored, and the best improving move is immediately applied. Passes stop when no move improves the objective or the search budget expires. This is FM-inspired boundary refinement, but it is not claimed to reproduce the classical Fiduccia--Mattheyses bucket/locking algorithm exactly.

\subsection{The affected-hyperedge locality theorem}
The key speedup follows from a simple dependency argument. Moving neuron $v$ changes exactly one label, $p(v)$. Which source-rooted multicast objects can notice this change?

\begin{proposition}[Affected-source locality]
For a move of neuron $v$, the only source-rooted hyperedges whose destination sets or route unions can change are
\begin{equation}
\mathcal{A}(v)=\{v\}\cup\{u:(u,v)\in E\}=\{v\}\cup N^-(v). \tag{19}
\end{equation}
\end{proposition}
\emph{Proof.} The hyperedge $h_v$ can change because its source core $p(v)$ changes. For another source $u\neq v$, the source core $p(u)$ is unchanged. Its destination-core set depends only on labels of neurons in $N^+(u)$. If $(u,v)\notin E$, then $v\notin N^+(u)$, so none of those labels changes and $D_u$ is identical. Therefore its route union and link-load contribution are identical as well. \hfill$\square$

Figure~4 visualizes this dependency footprint. The observation is elementary, but it converts a global-looking multicast objective into a local move evaluation.

\subsection{Exact incremental gain evaluation}
For every source $u$, the cached state stores its current source core, distinct remote destination-core set, routed-link set, weighted hop contribution, and link-load contribution. Global counters store total routed hops, per-link loads, per-core neuron counts, and per-core activity.

For a tentative move $v:a\rightarrow b$:
\begin{enumerate}
\item compute $\mathcal{A}(v)$;
\item subtract old cached route contributions of affected sources from global link loads;
\item temporarily change $p(v)$;
\item recompute only affected destination sets and route unions;
\item add the new affected contributions and evaluate the full objective;
\item either commit the move or restore the state.
\end{enumerate}

\begin{figure*}[t]
\centering
\includegraphics[width=0.96\textwidth]{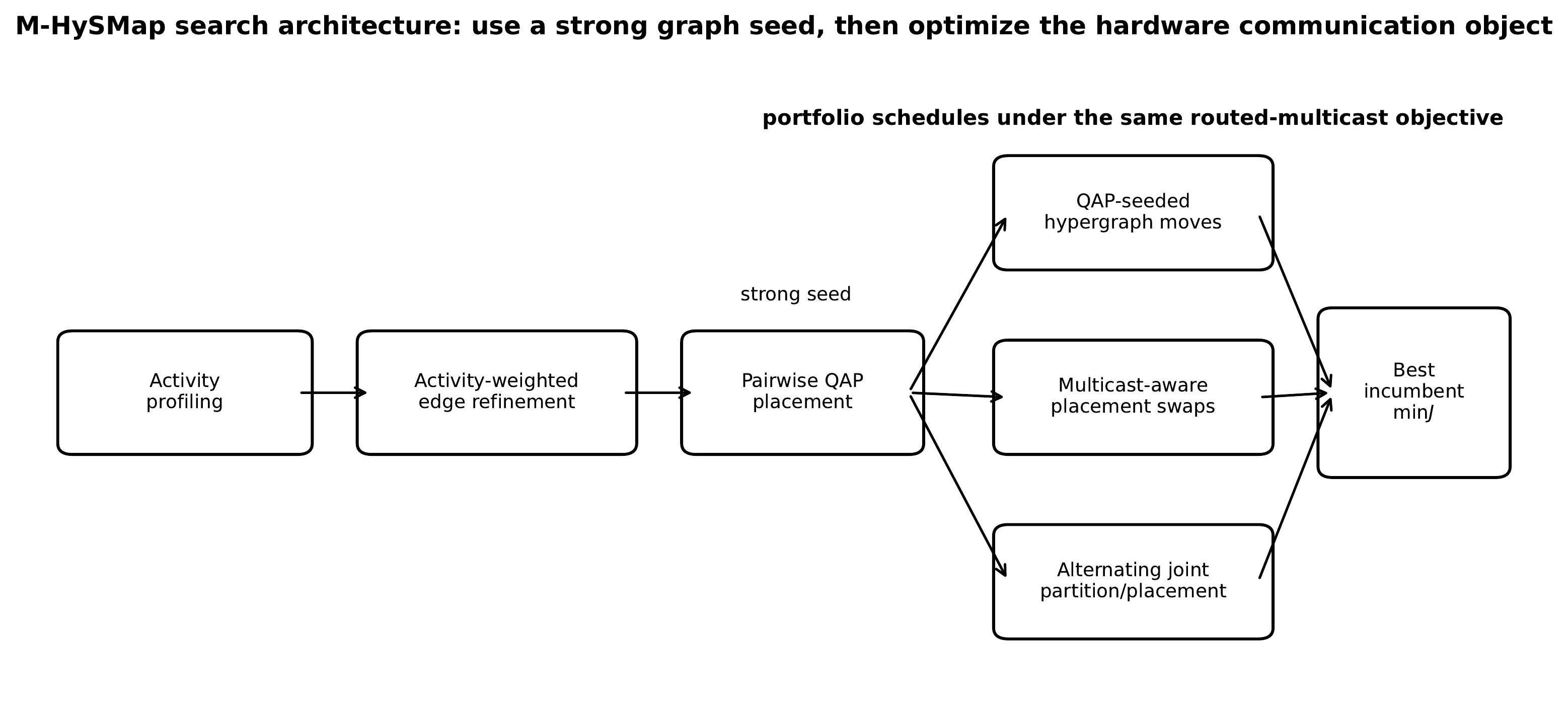}
\caption{Search architecture. Strong activity-aware graph/QAP seeds are refined by route-aware neuron moves, multicast placement swaps, and alternating joint schedules. The portfolio returns the minimum-objective incumbent among its candidates.}
\end{figure*}

\begin{figure}[t]
\centering
\includegraphics[width=0.98\columnwidth]{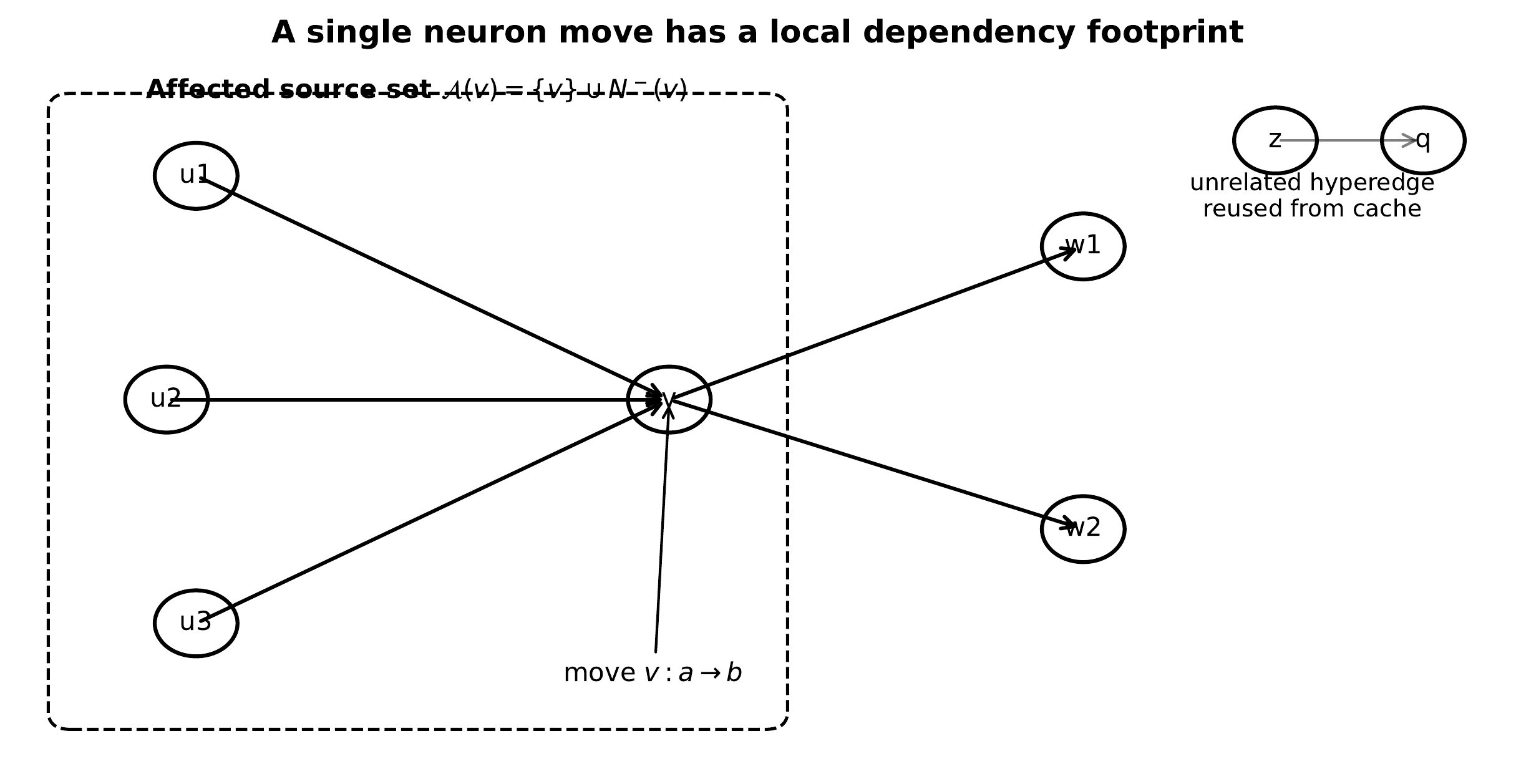}
\caption{Moving $v$ affects $h_v$ and predecessor-rooted hyperedges. Unrelated source contributions are provably unchanged and remain cached.}
\end{figure}

\refstepcounter{proposition}\noindent\textbf{Proposition \theproposition\ (Exactness of the incremental evaluator).} \emph{Assume the cached state equals the full objective decomposition before a tentative move. Replacing the old contributions of exactly $\mathcal{A}(v)$ with their recomputed post-move contributions yields the same objective value as a complete recomputation from scratch.}
\emph{Proof.} By Proposition 1, every source outside $\mathcal{A}(v)$ has an unchanged source core, destination-core set, route union, and link-load contribution. Their cached contributions therefore equal their full-recompute contributions after the move. The algorithm explicitly removes and recomputes every source in $\mathcal{A}(v)$ and then recomputes the global max-load and variance terms from the resulting load counters. Thus the aggregate state is identical to the full decomposition, up to floating-point arithmetic. \hfill$\square$

The implementation also performs a full recomputation at the end of refinement and records the absolute objective discrepancy as a validation check. The largest reported error is $2.3\times10^{-10}$ (Table~4).

\begin{figure}[t]
\small
\hrule\vspace{3pt}
\textbf{Algorithm 1} Incremental multicast partition refinement\\[-1pt]
\textbf{Require:} partition $p$, placement $\pi$, rates $r$, balance slack $s$\\
\begin{tabular}{@{}r@{\ }p{0.85\columnwidth}@{}}
1: & initialize cached routes and global link loads\\
2: & \textbf{for each} refinement pass \textbf{do}\\
3: & \quad $B\leftarrow$ shuffled boundary nodes\\
4: & \quad \textbf{for} $v\in B$ \textbf{do}\\
5: & \qquad construct feasible candidate core set $C(v)$\\
6: & \qquad $A\leftarrow\{v\}\cup N^-(v)$\\
7: & \qquad \textbf{for} $b\in C(v)$ \textbf{do}\\
8: & \qquad\quad subtract cached contributions of $A$\\
9: & \qquad\quad tentatively set $p(v)\leftarrow b$\\
10: & \qquad\quad recompute routes only for sources in $A$\\
11: & \qquad\quad evaluate complete $\mathcal{J}$ and restore tentative state\\
12: & \qquad \textbf{end for}\\
13: & \qquad commit the best strictly improving feasible move, if any\\
14: & \quad \textbf{end for}\\
15: & \quad \textbf{if} no move was committed \textbf{then}\\
16: & \qquad \textbf{break}\\
17: & \quad \textbf{end if}\\
18: & \textbf{end for}
\end{tabular}
\vspace{3pt}\hrule
\end{figure}

\subsection{What the complexity reduction really is}
A full candidate evaluation must rebuild destination-core sets and route unions across the SNN. Abstractly, if $H=R+C-2$ is the mesh diameter, a global evaluation is proportional to the work required to scan all fanouts and form all source route unions, plus $O(M)$ work to summarize link loads.

The incremental evaluator replaces the global source scan with only $u\in\mathcal{A}(v)$, plus the same small mesh-link summary. A useful qualitative expression is
\begin{align}
T_{\mathrm{inc}}(v)&=O(M)+O\!\left(\sum_{u\in\mathcal{A}(v)}\bigl(\deg^+(u)+|D_u|H\bigr)\right), \tag{20}\\
T_{\mathrm{full}}&=O(M)+O\!\left(\sum_{u\in V}\bigl(\deg^+(u)+|D_u|H\bigr)\right). \tag{21}
\end{align}
This is not a claim that every implementation operation is constant-time; for example, link-load counters must still be copied or summarized. The structural saving is that unrelated source-rooted multicast objects are never re-routed. Empirically, only 2.32--2.71 sources are affected on average in the reported configurations, which aligns with the observed 4.7--12.7$\times$ speedup.

\subsection{Multicast-aware placement neighborhood}
For fixed partition $p$, the placement neighborhood consists of 2-swaps of logical core locations. Starting from a permutation $\pi$, a candidate swap of cores $i$ and $j$ is accepted when it strictly decreases the same multicast objective $\mathcal{J}(p,\pi)$. The search uses several initial permutations: strong pairwise-QAP placements first, then row-major and randomized restarts if budget remains.

This is an anytime heuristic: at every point it has a feasible incumbent, and additional time permits more starting points and swaps. For very small $k$, exhaustive permutation enumeration is available to certify the placement optimum in a symmetry-normalized space; this exact path is a diagnostic feature rather than the mode used in the main experiments.

\subsection{Alternating joint refinement}
Partition and placement influence one another, so M-HySMap also alternates them. A joint cycle:
\begin{enumerate}
\item improves placement under the multicast objective, seeded by the current placement and a fresh pairwise-QAP proposal;
\item refines the partition with incremental multicast-aware neuron moves under that physical placement;
\item preserves the best $(p,\pi)$ encountered.
\end{enumerate}
A final placement polish is applied to the best partition. The heavy runs use two joint cycles.

\subsection{Portfolio incumbent selection}
Different local-search schedules can settle in different basins. M-HySMap therefore evaluates a compact portfolio including QAP-seeded hypergraph refinement, multicast-aware placement from strong seeds, and joint partition-placement schedules with different emphasis. Let $\mathcal{P}$ be the set of resulting candidates. The final result is simply
\begin{equation}
(p^\star,\pi^\star)=\arg\min_{(p,\pi)\in\mathcal{P}}\mathcal{J}(p,\pi). \tag{22}
\end{equation}

\begin{proposition}[Safe-best property]
The selected portfolio objective is no worse than the objective of any candidate explicitly included in $\mathcal{P}$.
\end{proposition}
\emph{Proof.} Immediate from the definition of the minimum in Eq.~(22). \hfill$\square$

The property is modest but useful: the portfolio is a robustness device, not a probabilistic claim of global optimality. Because every accepted local move strictly reduces $\mathcal{J}$ and the feasible mapping space is finite, an unbudgeted fixed schedule cannot cycle; it terminates when its examined neighborhood contains no improving move.

\begin{table}[t]
\centering
\caption{Benchmark configurations. The $7\times7$ case is a smaller stress run.}
\footnotesize
\begin{tabular}{lrrrr}
\toprule
Config & Cores & Seeds & Neurons & Synapses\\
\midrule
0.0015, $4\times4$ & 16 & 20 & 79 & $238\pm12$\\
0.002, $4\times4$ & 16 & 20 & 108 & $442\pm19$\\
0.002, $5\times5$ & 25 & 20 & 108 & $442\pm19$\\
0.002, $6\times6$ & 36 & 20 & 108 & $442\pm19$\\
0.002, $7\times7$ & 49 & 5 & 108 & $440\pm6$\\
0.003, $4\times4$ & 16 & 10 & 163 & $1028\pm36$\\
0.003, $5\times5$ & 25 & 10 & 163 & $1028\pm36$\\
0.003, $6\times6$ & 36 & 10 & 163 & $1028\pm36$\\
\bottomrule
\end{tabular}
\end{table}

\section{Experimental Methodology}
\subsection{Workloads and hardware abstraction}
We evaluate Potjans-inspired recurrent SNN workloads preserving layered excitatory/inhibitory population structure~[9]. The workloads are deliberately scaled so that many paired mapping runs can be completed under controlled conditions. Hardware is a two-dimensional mesh with deterministic XY shortest-path routing. Table~1 lists the configurations.

\subsection{Baselines and controlled ablation}
We compare:
\begin{itemize}
\item \textbf{Edge+QAP:} edge-cut partition refinement plus pairwise QAP placement;
\item \textbf{Activity+QAP:} source-activity-weighted edge refinement plus activity-weighted pairwise QAP placement;
\item \textbf{QAP-seeded M-HySMap:} multicast refinement initialized from Activity+QAP;
\item \textbf{M-HySMap:} the full portfolio incumbent.
\end{itemize}
All methods use the same generated SNN, profile, mesh, deterministic routing rule, and workload seed for a given job. This makes the progression a representation/search ablation: activity information is already present in the strongest graph baseline; M-HySMap's remaining gain comes from optimizing the source-rooted routed communication object and from the additional refinement neighborhoods.

The heavy configuration fixes two edge-refinement passes, four hypergraph-refinement passes, five placement restarts, 0.50 s per placement call, and two joint cycles. The evidence suite contains 115 completed jobs. The primary metric is routed multicast hops; maximum link load and the composite objective provide congestion context. Incremental speedup isolates the benefit of cached move evaluation rather than claiming equal end-to-end runtime against simpler baselines.

\section{Results}
\subsection{Routed multicast traffic}
Table~2 reports the main results. M-HySMap reduces routed multicast hops by 19.7--41.1\% relative to Edge+QAP and by 10.6--19.6\% relative to the stronger Activity+QAP baseline. The improvement persists across every reported configuration, including the $7\times7$ stress case with fewer seeds.

The controlled baseline is important for interpretation. Activity+QAP already knows which sources spike more often and already places high-flow core pairs close together. The additional gain therefore cannot be explained merely by adding activity weights or topology awareness. The algorithm changes what it counts: distinct destination-core delivery and shared routed links.

\subsection{Ablation: where the gain appears}
Averaged over configurations, Table~3 shows a normalized hop cost of 1.000 for Activity+QAP, 0.857 after QAP-seeded multicast refinement, and 0.848 for the final portfolio. Most of the improvement over the strongest graph baseline appears as soon as the search begins optimizing the multicast communication object; the portfolio provides a smaller robustness improvement on top.

The averaged maximum-link load is 20.5\% below Activity+QAP and 25.5\% below Edge+QAP. Because the objective explicitly includes worst-link and variance terms, the hop reduction is not obtained by blindly concentrating traffic onto a single link.

\subsection{Incremental update speed}
Table~4 compares the same partition-refinement logic under full objective recomputation and the incremental affected-hyperedge evaluator. Speedups range from 4.7$\times$ to 12.7$\times$. The average affected-source set remains small---roughly 2.3--2.7 source hyperedges per evaluated boundary node---while the full objective validation error stays at floating-point scale.

\begin{table*}[t]
\centering
\caption{Main routed multicast hop results. Values are mean $\pm$ standard deviation. Improvements are hop reductions achieved by M-HySMap.}
\small
\begin{tabular}{lrrrrrr}
\toprule
Config & Seeds & Edge+QAP & Activity+QAP & M-HySMap & vs Edge & vs Activity\\
\midrule
0.0015, $4\times4$ & 20 & $1922\pm775$ & $1409\pm643$ & $1133\pm563$ & 41.1\% & 19.6\%\\
0.002, $4\times4$ & 20 & $8642\pm3618$ & $7332\pm3341$ & $5934\pm2677$ & 31.3\% & 19.1\%\\
0.002, $5\times5$ & 20 & $10389\pm4210$ & $8535\pm3675$ & $7309\pm3161$ & 29.6\% & 14.4\%\\
0.002, $6\times6$ & 20 & $13264\pm5563$ & $10624\pm4990$ & $9499\pm4404$ & 28.4\% & 10.6\%\\
0.002, $7\times7$ & 5 & $13819\pm4608$ & $11079\pm3489$ & $9603\pm2974$ & 30.5\% & 13.3\%\\
0.003, $4\times4$ & 10 & $113919\pm17494$ & $102637\pm16546$ & $84816\pm14631$ & 25.5\% & 17.4\%\\
0.003, $5\times5$ & 10 & $147321\pm20830$ & $135338\pm20194$ & $116539\pm17280$ & 20.9\% & 13.9\%\\
0.003, $6\times6$ & 10 & $181047\pm26287$ & $167462\pm24468$ & $145383\pm21471$ & 19.7\% & 13.2\%\\
\bottomrule
\end{tabular}
\end{table*}

\begin{figure*}[t]
\centering
\includegraphics[width=0.90\textwidth]{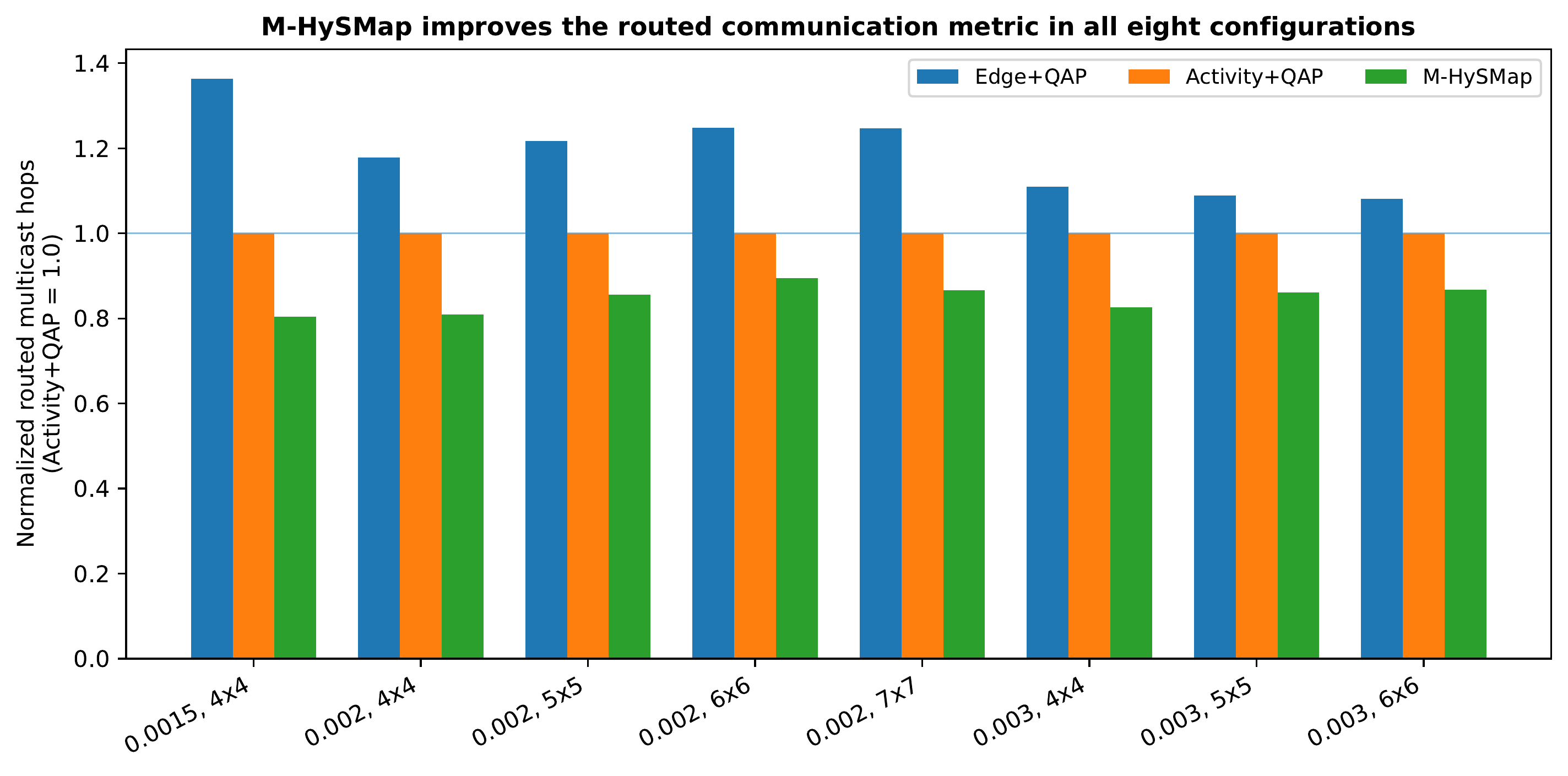}
\caption{Normalized routed multicast hops. Activity+QAP is 1.0 in each configuration; M-HySMap improves the routed metric in all eight configurations.}
\end{figure*}

\begin{table}[t]
\centering
\caption{Compact ablation summary averaged over configurations.}
\footnotesize
\setlength{\tabcolsep}{4pt}
\begin{tabular}{@{}lrrrr@{}}
\toprule
Method & Norm. & Hops & Obj. & Max link\\
\midrule
Edge+QAP & 1.192 & 61290 & 65458 & 1986.1\\
Activity+QAP & 1.000 & 55552 & 59391 & 1860.8\\
QAP-seeded M-HySMap & 0.857 & 47653 & 50085 & 1523.7\\
M-HySMap & 0.848 & 47527 & 49824 & 1479.3\\
\bottomrule
\end{tabular}
\end{table}

\section{Related Work and Positioning}
\subsection{SNN mapping and neuromorphic NoCs}
Existing SNN mapping flows commonly decompose the problem into clustering/partitioning followed by placement. SpiNeMap clusters and places SNN communication to reduce spike latency and energy~[1]. SNEAP targets NoC-based neuromorphic platforms and reduces average hop count~[7]. NeuToMa emphasizes topology-aware mapping and workload balance~[13], while NeuMap models communication patterns and uses metaheuristic placement for edge-AI workloads~[12]. These approaches motivate our strongest graph baseline: activity-aware partitioning plus topology-aware pairwise placement.

Ronzani and Silvano recently make the closely related case that SNN mapping should be raised from graphs to source-rooted hypergraphs~[14]. Their work studies hyperedge overlap/locality and redesigns partitioning and placement around the hypergraph abstraction. M-HySMap should therefore not be read as claiming that hypergraphs themselves are new to SNN mapping. The distinction in this work is the routed objective and the corresponding search machinery: destination-core sets are translated into unions of deterministic mesh routes; worst-link and load-variance terms are explicit; strong Activity+QAP seeds are refined under that route-aware objective; and move gains exploit the exact affected-source dependency set.

\subsection{Graph and hypergraph partitioning}
Kernighan--Lin and Fiduccia--Mattheyses are foundational local-refinement methods for graph partitioning~[5, 3]. Hypergraphs are natural when a communication object connects one source or net to many endpoints and have long been used in VLSI partitioning and sparse matrix decomposition~[4, 2, 11]. M-HySMap borrows the general local-refinement spirit but specializes the objective and dependency structure to source-rooted spike delivery on a routed mesh.

\subsection{Pairwise placement versus route-union placement}
The QAP abstraction matches logical communication flow to physical distance~[6]. In M-HySMap it serves two roles: as a strong graph baseline and as a seed generator. Its limitation is representational rather than algorithmic: $\sum_{ij}F_{ij}D_{\pi(i),\pi(j)}$ counts pairwise flows, whereas the hardware multicast cost can share links across multiple destinations of the same source. The two objectives are therefore correlated but not equivalent.

\section{Limitations and What the Results Do Not Claim}
The current evidence uses scaled Potjans-inspired recurrent workloads rather than a broad standard neuromorphic benchmark suite. The networks contain 79--163 neurons and approximately 238--1028 synapses, chosen to support controlled paired experiments across many seeds and mapper schedules. Larger full-system SNNs remain important future work.

The reported metrics are mapping-level traffic proxies. We do not convert routed-hop reductions into energy, latency, throughput, or cycle claims without a calibrated target architecture. Routing is deterministic XY; adaptive routing or hardware multicast policies would change the route-set operator in Eq.~(7). The formulation is deliberately modular: $T_u$ can be replaced by the route set induced by another deterministic architecture policy, and the objective can include additional calibrated terms.

Activity rates are obtained from a finite profiling interval and may vary across application phases. A natural extension is scenario-aware mapping with several activity vectors and robust or expected-cost objectives.

Finally, the search is heuristic. Greedy neuron moves and 2-swap placement searches guarantee monotone improvement along the moves they accept, not global optimality. The portfolio reduces sensitivity to one search trajectory but does not eliminate local minima. Time-budgeted placement also means exact incumbents can depend mildly on machine speed; deterministic evaluation-count budgets would improve archival reproducibility.

\begin{table*}[t]
\centering
\caption{Runtime benefit of exact incremental hyperedge updates.}
\small
\begin{tabular}{lrrrrr}
\toprule
Config & Full recompute (s) & Incremental (s) & Speedup & Avg. affected & Max error\\
\midrule
0.0015, $4\times4$ & 0.438 & 0.061 & 7.2$\times$ & 2.48 & $1.8\times10^{-12}$\\
0.002, $4\times4$ & 1.044 & 0.113 & 9.2$\times$ & 2.71 & $9.1\times10^{-12}$\\
0.002, $5\times5$ & 1.676 & 0.356 & 4.7$\times$ & 2.60 & $8.2\times10^{-12}$\\
0.002, $6\times6$ & 2.462 & 0.431 & 5.7$\times$ & 2.51 & $2.0\times10^{-11}$\\
0.002, $7\times7$ & 1.747 & 0.351 & 5.0$\times$ & 2.65 & $1.1\times10^{-11}$\\
0.003, $4\times4$ & 11.844 & 0.933 & 12.7$\times$ & 2.41 & $7.3\times10^{-11}$\\
0.003, $5\times5$ & 13.311 & 1.072 & 12.4$\times$ & 2.56 & $2.3\times10^{-10}$\\
0.003, $6\times6$ & 13.995 & 1.135 & 12.3$\times$ & 2.32 & $1.5\times10^{-10}$\\
\bottomrule
\end{tabular}
\end{table*}

\begin{figure*}[t]
\centering
\includegraphics[width=0.90\textwidth]{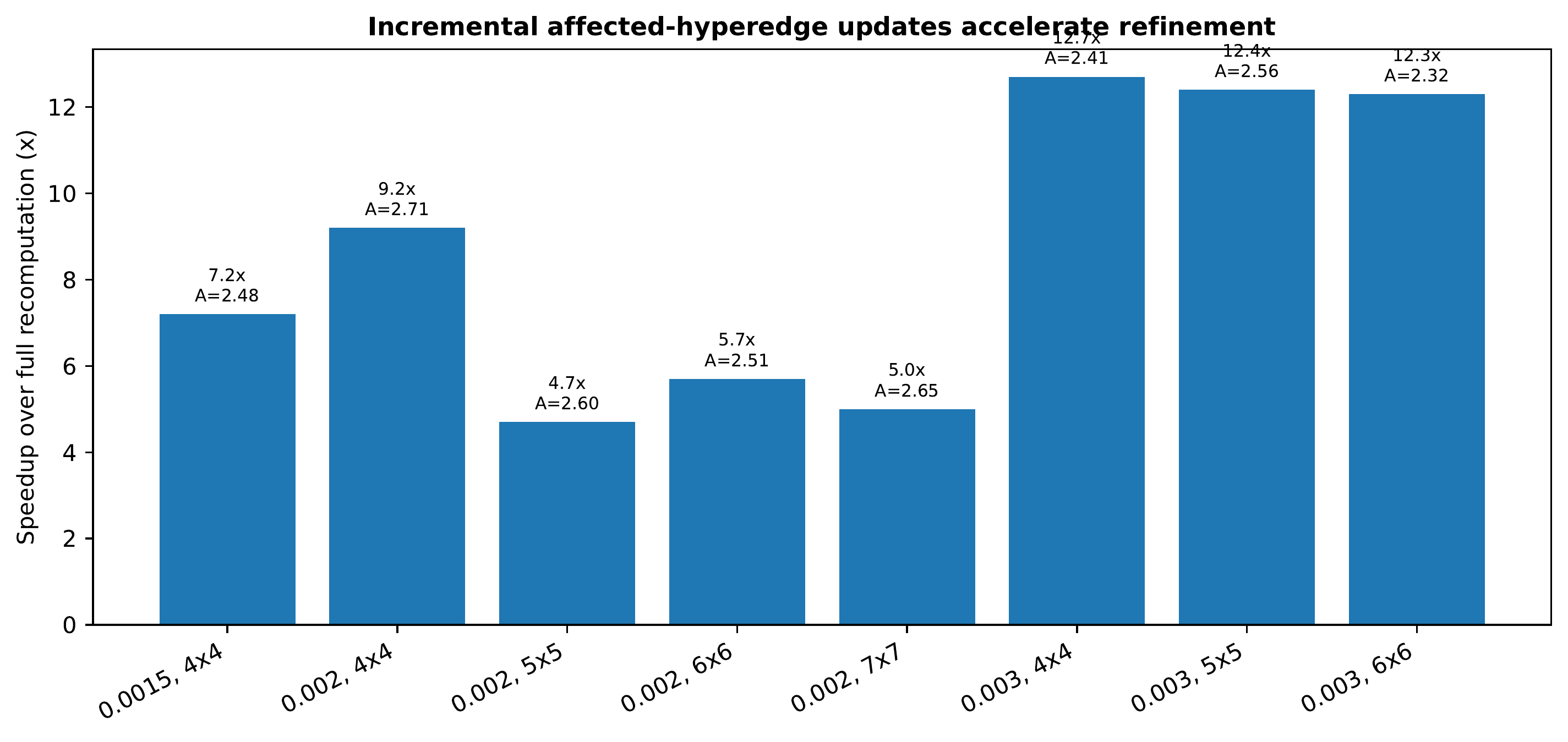}
\caption{Incremental update speedup. Labels also report the average number $A$ of affected source hyperedges, illustrating why locality is valuable.}
\end{figure*}

\section{Conclusion}
The main lesson is an abstraction lesson. SNN mapping is not only a problem of assigning weighted graph edges to nearby cores. A source spike is a one-to-many communication event, and its physical cost depends on distinct destination cores and the union of links used to reach them. Modeling each source as an activity-weighted directed hyperedge makes that structure explicit; optimizing the routed union turns it into a hardware-facing objective.

The algorithmic consequence is equally useful. Although the multicast objective looks global because it aggregates mesh congestion, one neuron move has a sharply local source dependency: $h_v$ and predecessor-rooted hyperedges only. Caching the rest converts that structural fact into exact incremental gains and 4.7--12.7$\times$ faster refinement. Strong Activity+QAP seeds, route-aware partition moves, placement swaps, and safe-best portfolio selection then provide a practical combinatorial search strategy.

Across the 115-job evidence suite, the resulting mappings reduce routed multicast hops by 10.6--19.6\% relative to the strongest activity-aware graph/QAP baseline while also lowering average worst-link load. The broader implication is that, for event-driven many-core mapping, choosing a communication object that matches the hardware semantics can matter as much as improving the optimizer applied to an older abstraction.

\section*{Reproducibility Notes}
The implementation is organized around fixed workload seeds, JSON experiment configurations, subprocess-isolated jobs, and per-run CSV logging. The heavy evidence configuration records profiling time, input rate, partition passes, placement restarts, placement budgets, joint cycles, and balance slack. A public code release should contain only the source and the exact evidence configuration/result files needed to reproduce the preprint, excluding stale development archives and intermediate logs not required for compilation or execution.


\begin{thebibliography}{99}\small
\bibitem{balaji2020} A. Balaji et al., ``Mapping Spiking Neural Networks to Neuromorphic Hardware,'' \emph{IEEE Trans. VLSI Systems}, vol. 28, no. 1, pp. 76--86, 2020. doi:10.1109/TVLSI.2019.2951493.

\bibitem{catalyurek1999} Ü. V. Çatalyürek and C. Aykanat, ``Hypergraph-Partitioning-Based Decomposition for Parallel Sparse-Matrix Vector Multiplication,'' \emph{IEEE TPDS}, vol. 10, no. 7, pp. 673--693, 1999. doi:10.1109/71.780863.

\bibitem{fiduccia1982} C. M. Fiduccia and R. M. Mattheyses, ``A Linear-Time Heuristic for Improving Network Partitions,'' in \emph{Proc. DAC}, pp. 175--181, 1982. doi:10.1145/800263.809204.

\bibitem{karypis1999} G. Karypis and V. Kumar, ``Multilevel $k$-Way Hypergraph Partitioning,'' in \emph{Proc. DAC}, pp. 343--348, 1999. doi:10.1145/309847.309954.

\bibitem{kernighan1970} B. W. Kernighan and S. Lin, ``An Efficient Heuristic Procedure for Partitioning Graphs,'' \emph{Bell System Technical Journal}, vol. 49, no. 2, pp. 291--307, 1970. doi:10.1002/j.1538-7305.1970.tb01770.x.

\bibitem{koopmans1957} T. C. Koopmans and M. Beckmann, ``Assignment Problems and the Location of Economic Activities,'' \emph{Econometrica}, vol. 25, no. 1, pp. 53--76, 1957. doi:10.2307/1907742.

\bibitem{li2020} S. Li et al., ``SNEAP: A Fast and Efficient Toolchain for Mapping Large-Scale Spiking Neural Network onto NoC-Based Neuromorphic Platform,'' arXiv:2004.01639, 2020.

\bibitem{maass1997} W. Maass, ``Networks of Spiking Neurons: The Third Generation of Neural Network Models,'' \emph{Neural Networks}, vol. 10, no. 9, pp. 1659--1671, 1997. doi:10.1016/S0893-6080(97)00011-7.

\bibitem{potjans2014} T. C. Potjans and M. Diesmann, ``The Cell-Type Specific Cortical Microcircuit: Relating Structure and Activity in a Full-Scale Spiking Network Model,'' \emph{Cerebral Cortex}, vol. 24, no. 3, pp. 785--806, 2014. doi:10.1093/cercor/bhs358.

\bibitem{roy2019} K. Roy, A. Jaiswal, and P. Panda, ``Towards Spike-Based Machine Intelligence with Neuromorphic Computing,'' \emph{Nature}, vol. 575, pp. 607--617, 2019. doi:10.1038/s41586-019-1677-2.

\bibitem{schlag2023} S. Schlag et al., ``High-Quality Hypergraph Partitioning,'' \emph{ACM Journal of Experimental Algorithmics}, vol. 27, 2023. doi:10.1145/3529090.

\bibitem{xiao2022neu} C. Xiao, J. Chen, and L. Wang, ``Optimal Mapping of Spiking Neural Network to Neuromorphic Hardware for Edge-AI,'' \emph{Sensors}, vol. 22, no. 19, 7248, 2022. doi:10.3390/s22197248.

\bibitem{xiao2022topo} C. Xiao, Y. Wang, J. Chen, and L. Wang, ``Topology-Aware Mapping of Spiking Neural Network to Neuromorphic Processor,'' \emph{Electronics}, vol. 11, no. 18, 2867, 2022. doi:10.3390/electronics11182867.

\bibitem{ronzani2026} M. Ronzani and C. Silvano, ``A Case for Hypergraphs to Model and Map SNNs on Neuromorphic Hardware,'' arXiv:2601.16118, 2026.
\end{thebibliography}
\end{document}